\documentclass[11pt]{article}
\usepackage{paper}
\usepackage{amsmath,amssymb,amsfonts,latexsym,url,xspace,algorithmic}
\usepackage{fullpage}
\newcommand{\maxtwo}{\ensuremath{\mathop{\rm max2}}}
\newcommand{\pro}{\ensuremath{\mathrm{rev}}}
\newcommand{\supp}{\ensuremath{\mathrm{supp}}}
\newcommand{\B}{\ensuremath{\mathcal{B}}}

\title{Send Mixed Signals -- Earn More, Work Less}

\author{
Peter Bro Miltersen\footnote{The author acknowledge support from
the Danish National Research Foundation and The National Science Foundation of
China (under the grant 61061130540) for the Sino-Danish Center for the Theory of
Interactive Computation and from the Center for Research in Foundations of Electronic Markets (CFEM), supported by the Danish Strategic Research Council.}
\\ Aarhus University
\and
Or Sheffet
\\ Carnegie Mellon University
}

\date{}

\begin{document}
\maketitle

\begin{abstract}
Emek {\em et al.} presented a model of probabilistic single-item second price auctions where an auctioneer who is informed about the type of an item for sale, broadcasts a signal about this type to uninformed bidders. They proved that finding the optimal (for the purpose of generating revenue) {\em pure} signaling scheme is strongly NP-hard. In contrast, we prove that finding the optimal {\em mixed} signaling scheme can be done in polynomial time using linear programming. For the proof, we show that the problem is strongly related to a problem of optimally bundling divisible goods for auctioning. We also prove that a mixed signaling scheme can in some cases generate twice as much revenue as the best pure signaling scheme and we prove a generally applicable lower bound on the revenue generated by the best mixed signaling scheme.
\end{abstract}

\section{Introduction}
\label{sec:intro}

%Mechanism Design literature is replete with various auction settings, BLAH BLAH%BLAH... Single parameters have some fairly reasonable approximations, whereas %multi-parameters auctions are notoriously difficult, BLAH BLAH BLAH...
%
%I really don't think it is necessary to open this paper with lots of 
%general, vague motivational bullshit (GVMB). I dislike GVMB in general, but
%in particular for this paper: You here provide a very well-defined 
%technical improvement to previous work and GVMB may cause the readers 
%to miss this point.  The paper is good if and only if people think 
%the previous work was good. So I think one should plunge right in.

%I made it clear that it is part of the fixed setup that a 2nd price auction is
%conducted, *and that the auctioneer cannot change this* by moving the second price auction description into description of the setup. Otherwise we invite the reader to already at this point realize that
%it is pretty stupid to conduct a second price auction under the stated 
%assumptions; if the auctioneer can change that bit, she should simply post
%a price! This also aligns well with the ``defense'' of Feldman you mentioned: 
%it is part of the rules that a 2nd price auction must be performed 
%and this is this natural setup that we will have to deal with.
\paragraph{Background} Emek {\em et al.} \cite{Emek} recently introduced the following {\em probabilistic single-item auction} model: An auctioneer sells a single item to $n$ bidders. The item comes from one of $m$ different types, and the valuations of the bidders for the item vary between the different $m$ types, with the valuation $v_{ij}$ of bidder $i$ for an item of type $j$ being common knowledge (or at least known to the auctioneer). The actual type of the item is determined by nature, with the probability $p_j$ of each type $j$ occurring also being common knowledge. There is asymmetry of information in the setting in one respect only: The auctioneer knows the \emph{realization} of the type of the item, whereas the bidders do not. The auction proceeds by the auctioneer broadcasting to the bidders a single {\em signal} about the type of the item. In the work of Emek {\em et al.}, the signaling schemes considered are {\em pure}. That is, the signal is simply some function of the type of the item and in particular, there is a one-to-one correspondence between signaling schemes and partitions of the set of types. After receiving the signal, the bidders bid for the item in a standard 2nd price sealed-bid auction. It is assumed that bidders are risk neutral and play the dominant strategy of bidding their expected valuation given their signal in this auction. Emek {\em et al.} investigated the following question: {\em To which extent can the auctioneer exploit her informational advantage to increase revenue by choosing the signaling scheme appropriately?}

Emek {\em et al.} show examples where non-trivial schemes significantly outperform the two trivial ones (which are: fully revealing the type of the item and revealing nothing at all). They show that it is strongly NP-hard to compute the pure signaling scheme that maximizes revenue among all such schemes. Their main result is a polynomial time algorithm that finds a pure signaling scheme that approximates the revenue of the optimal one within a constant factor.  

%So using this integral signaling scheme causes the auctioneer to gain $1/2$ in the second-price auction for item of type $i$, and a gain of $1/2$ overall. Emek et al's main result is to devise in poly-time, some integral signaling scheme, which is $O(1)$-approximation to optimal signaling scheme.

%I really think that the ``Earning more'' aspect should be promoted to be a result, rahter than just an example and a figure in the intro.
\paragraph{Our Results} In this work, we consider the extension of the model of Emek {\em et al.} consisting of allowing the auctioneer to use a {\em mixed} signaling scheme. In such
a scheme, the auctioneer, after witnessing the realization of the item, picks a signal at random according to some probability distribution depending on this realization. We show that by making this very natural extension of the model, we kill two birds with one stone: 
\begin{itemize}
\item{}{\em We earn more:} We show that there are problem instances (with arbitrarily many bidders) where the optimal mixed signaling scheme generates twice the revenue generated by the optimal pure signaling scheme. Also, we show that the revenue generated is never less than $\B/2$, where $\B = \min_{i'} \left(\sum_j \max_{i\neq i'} p_j v_{i,j}\right)$. We postpone to Section~\ref{sec:approximating-benchmark} a detailed discussion as to why this particular benchmark is meaningful.
\item{}{\em We work less:} We show that the optimal mixed signaling scheme can be found in polynomial time, by devising a concise linear program describing this optimal scheme. While it is certainly intuitive that linear programming should be used to find an optimal mixed strategy, we need to prove several structural results concerning the optimal solution before being able to devise a polynomial sized linear program in the present setting. 
\end{itemize}

\paragraph{Discussion of the model} We are aware that in the setting of Emek {\em et al.} (which is our setting as well), having the valuations known to the auctioneer makes it is less than obvious why the model requires the item to be sold in a $2$nd price auction. Indeed, simply posting an appropriately chosen price would generate more revenue. Also, the assumption about valuations being known to the auctioneer is itself questionable (note in particular that there is no obvious way to truthfully elicit these valuations from the bidders). To address this critique, we note that Emek {\em et al.} use the complete information setup and the associated results outlined above as a component in an analysis of a {\em Bayesian} variant of the setup, where the auctioneer is unaware of the actual valuations and has to base her signaling scheme solely on a probabilistic model thereof. Our mixed signaling variant can replace the original pure one also in this Bayesian variant and will increase its revenue and decrease its computational complexity. (We believe it would be interesting to understand how well such a scheme approximates the revenue of the {\em optimal} Bayesian auction in the sense of Myerson \cite{Myer} in this setting, and suggest this as a possible topic for future work.) Another, more down-to-earth answer to the critique is that a 2nd price auction is simply a very natural, well-known and wide-spread scheme for selling an item and that it therefore makes sense to fix this part of the setup when the main agenda is to investigate how signaling can improve revenue. In essence, our setting allows us to give an \emph{exact quantification} of the gain the auctioneer can obtain by optimally leveraging her informational advantage. And, as discussed in Section~\ref{sec:approximating-benchmark}, if the valuations are not dominated by a single bidder, then our benchmark-approximation analysis shows that the revenue from $2$nd price auctions is comparable with the revenue of the posted price scheme.

\paragraph{Related Research} Due to our setup being a variant of the setup of Emek {\em et al.}, we refer to their paper for an extensive discussion regarding  works dealing with sellers exploiting their informational advantage (dating back to the Nobel Prize winning work of Akerlof~\cite{akerlof1970}). However, unlike their pure-signal scheme, our mixed-signal scheme has an alternative interpretation as a model in which the auctioneer sells $m$ \emph{divisible} goods to $n$ bidders which have simple linear valuations per item, by \emph{bundling} subsets of these goods together (see Section~\ref{subsec:divisible_goods_model}). The problem of bundling goods, including divisible goods, has received considerable attention in the economics literature (e.g., \cite{Adam}) as well as the problem of auctioning divisible goods (see~\cite{Back_Zender_2001, Ausubel04auctioningmany, Iyengar_Kumar_2008} and the books ~\cite{cramton2010combinatorial, klemperer2004auctions}). However, our particular model does not seem to have been considered. 

\paragraph{Organization of Paper} First, in Section~\ref{sec:preliminaries}, we provide the details of our mixed-signals model and demonstrate that it is equivalent to an auction model concerning bundling of divisible goods.
%Also in Section~\ref{sec:preliminaries}, we survey known work that deals with single and multi-parameter auctions, and with bundling.
In Section \ref{sec:earnmore}, we present the examples where sending mixed signals significantly increases revenue.
Then, in Section~\ref{sec:opt-mixed-singals-LP} we show that it is feasible to devise the optimal mixed-signals scheme in poly-time, using a polynomial size LP. Finally, we show in Section~\ref{sec:approximating-benchmark} that the revenue of the mixed signal auction is at least half the benchmark $\B$. We conclude with discussion and open problems in Section~\ref{sec:conclusion}.

\section{Preliminaries -- The Model}
\label{sec:preliminaries}
\subsection{The Problem Formalization}
\label{subsec:problem_formally}
In a {\em probabilistic single-item auction with mixed signals}, an auctioneer wants to sells an item drawn from a known distribution $(p_1,p_2,\ldots,p_m)$ over $m$ types. There are $n$ bidders that wish to purchase the item, each with valuation $v_{i,j}$ for an item of type $j$, with these valuations being common knowledge. The auctioneer observes the type of the item and broadcasts a signal to the bidders. The signaling scheme is strategically chosen by the auctioneer in advance and is given by a map $\varphi:[m] \times \mathbb{N} \rightarrow [0,1]$, such that for every $j$, the auctioneer declares signal $S$ with probability $\varphi(j,S)$. As we later show, the overall number of signals we send can be assumed to be finite in the scheme generating the largest revenue, so we can assume that from some signal $B$ and onwards, the function $\varphi$ is identically $0$ (formally, $\forall j$ and $\forall S\geq B, \ \varphi(j,S)=0$). We abuse notation and identify $S$ with its support (i.e., the set $\{j:\ \varphi(j,S)>0\}$). We also alternate between the notations $\varphi(j,S)$ and $\varphi_{j,S}$. We denote $\mathcal{S}_\varphi$ as the set of all possible signals, i.e. $\mathcal{S}_\varphi = \{S:\ \varphi(j,S) > 0 \textrm{ for some }j\}$.
After receiving the signal, the bidders participate in a 2nd price auction for the item. 

A {\em pure} signaling scheme is one where $\varphi(j,S) \in \{0,1\}$ for all $j,S$. The variant of the above setup where the autioneer is restricted to use a pure signaling $\varphi$ is the
probabilistic single-item auction with pure signals originally suggested by Emek {\em et al.} Let us repeat a derivation from Emek et al. for the more general mixed case. For a fixed signal $S$, the probability of the auctioneer broadcasting this signal is $\sum_j p_j \varphi(j,S)$, and so, given that the auctioneer broadcasted the signal $S$, the probably that the item is of type $j$ is $\Pr[j|S] =  {p_j \varphi(j,S)}/ \left({\sum_{j'} p_{j'} \varphi(j',S)}\right)$. As a result, given signal $S$, the adjusted valuation of bidder $i$ over the item is $\E[v_i|S] = \sum_j \Pr[j|S]v_{i,j} =  {\sum_j v_{i,j} p_j \varphi(j,S)}/ \left({\sum_{j'} p_{j'} \varphi(j',S)}\right)$. We assume risk neutral bidders, who follow the dominant strategy of bidding this adjusted valuation in the 2nd price auction. Therefore, for signal $S$, the auctioneer's revenue is \[\maxtwo_i\left\{\E[v_i|S]  \right\} = \maxtwo_i \left\{ \frac {\sum_j v_{i,j} p_j \varphi(j,S)} {\sum_{j} p_{j} \varphi(j,S)} \right\}.\] We are interested in the  $\varphi$ that maximizes the expected revenue:
\begin{eqnarray*}
& \textrm{maximize } & \sum_{S\in\mathcal{S}_\varphi} \Pr[S]\maxtwo_i\left\{\E[v_i|S]  \right\} \cr
&&  = \sum_{S\in\mathcal{S}_\varphi} \maxtwo_i \left\{\sum_j  \left(v_{i,j}p_j\right)\varphi(j,S)\right\} = \sum_{S\in\mathcal{S}_\varphi} \maxtwo_i \left\{ \sum_j \psi_{i,j}\ \varphi(j,S)\right\} 
\end{eqnarray*}
where the last equality merely comes from introducing the definition $\psi_{i,j} = v_{i,j}p_j$. 
%Observe, as we assume full knowledge, $\psi_{i,j}$ can be immediately computed from the input.

\subsection{Equivalent Model of Divisible Goods}
\label{subsec:divisible_goods_model}
%As mentioned, the above derivation is part of the problem definition in Emek et al. However, as Emek et al considered only pure signaling schemes, they did not observe that this derivation converts the problem from having a prior over the item into the problem of \emph{bundling $m$ divisible goods} under a simple valuations model.

We observe that a probabilistic single-item auction with mixed signals can alternatively be seen as an auction where $m$ divisible goods are bundled and sold. The mixed signals are crucial for this characterization.
The alternative model may be defined as follows: The auctioneer wishes to sell $m$ heterogeneous divisible goods to $n$ bidders. She has $1$ unit of each of the goods (for example, she has $1$ kilogram from each of $m$ exotic spices). Each bidder $i$ has linear valuation of $\psi_{i,j}$ for each unit of good $j$, so bidder $i$ has a utility of $\sum_j x_j \psi_{i,j}$ if he receives $x_j$ units of each good $j$. The auctioneer sells her goods by \emph{bundling} several goods together. More precisely, she uses a bundling scheme $(\mathcal{S}, \phi)$, where in each bundle $S \in \mathcal{S}$, she places $\varphi_{j,S}$ units of good $j$, and then she runs a $2$nd price auction for each bundle. We assume that bidders follow their dominant strategy of bidding their valuation for the bundle for sale in each of these auctions. 

The analogy between signaling in the model of one good of $m$ different types, and bundling in the model of $m$ divisible goods, is clear. Given a probabilistic single-item auction with $n$ bidders and $m$ types, we can define a divisible goods auction with $n$ bidders and $m$ goods by letting $\psi_{i,j} = p_j v_{i,j}$. Conversely, given a divisible goods auction with $n$ bidders and $m$ goods, we can define a probabilisitc single-item auction with $n$ bidders and $m$ types by letting $(p_j)_{j=1}^m$ be an arbitrary probability distribution with $p_j > 0$ for each $j$ and letting $v_{i,j} = \psi_{i,j}/p_j$. Also, mixed signaling schemes in the probabilistic single-item auction and bundling schemes in the divisible goods auctions are syntactically the same objects. Finally,  it is readily checked that the expected revenue in the probabilistic single-item auction is identical (up to a scaling factor) to the revenue in the corresponding divisible good auction. Therefore, finding an optimal mixed signaling scheme in the first model is equivalent to finding an optimal bundling scheme in the latter.

As a result of the above, we allow ourselves the liberty to alternate between the two models.

\section{Earning More by Sending Mixed Signals}
\label{sec:earnmore}
A simple example where the best mixed signaling scheme outperforms the best pure signaling scheme is the following. Assume it is the case where $m=n=3$, the item is equally likely to be any one of the three types, and the valuations are the identity matrix (bidder $i$ wants only item of type $i$, so $v_{i,i} = 1$, and no other type, so $v_{i,j} = 0$ when $i\neq j$). A pure signaling scheme is forced to pair two of the three types, and results in expected revenue of $\frac{1}{3}$. In contrast, a mixed signaling scheme may use all $3$ signals $\{1,2\}, \{1,3\}, \{2,3\}$, and declare any signal that type $j$ belongs to with equal probability. (E.g., if the item type is $1$, then with probability $\frac{1}{2}$ the auctioneer declares $\{1,2\}$ and with probability $\frac{1}{2}$ she declares $\{1,3\}$.) Now, no matter what cluster $\{j, j'\}$ was declared, both bidder $j$ and bidder $j'$ know there's a $50\%$ chance that the item is of their desired type, resulting in a bid of $\frac{1}{2}$ from both bidder $j$ and bidder $j'$. Thus, the auctioneer gains revenue of $\frac{1}{2}$ with a mixed signaling scheme, exhibiting a gap of $1.5$ between the best mixed signaling scheme and the best pure signaling scheme. By a slightly more complex construction, we can get a gap of 2:
\begin{theorem}
For any even number $k$, there is a probabilistic single-item auction with $n=k+1$ bidders and $m=k+1$ types so that the optimal mixed signaling scheme has an expected revenue which is twice as big as that of the optimal pure signaling scheme.
\end{theorem}
\begin{proof}
Consider the auction with valuations as given in Figure~\ref{fig:fractional_vs_integral} and with nature choosing the type uniformly at random.
\begin{figure}
\begin{center}
\fbox{\parbox{4in}{
\begin{tabular*}{4in}{r|c|c|c|c|c|c|c}
& type $0$ & type $1$ & type $2$ & .. & .. & .. & type $k$ \cr 
\hline 
bidder $0$ & $m$ & 0 & 0 &  .. & .. & .. & 0 \cr 
bidder $1$ & 0 & 1 & 0 &  .. & .. & .. & 0 \cr 
bidder $2$ & 0 & 0 & 1 &  .. & .. & .. & 0 \cr 
. & & & & $\ddots$ & & & \cr
. & & & & & $\ddots$ & & \cr
. & & & & & & $\ddots$ & \cr
bidder $k$ & 0 & 0 & 0 &  .. & .. & .. & 1 \cr 
\end{tabular*}
}}
\end{center}
\caption{\label{fig:fractional_vs_integral} \scriptsize Valuations of an auction in which the optimal mixed signaling scheme has twice the revenue of the optimal pure signaling scheme.}
\end{figure}
A pure signaling scheme can only pair $k$ types, so in such a scheme, with probability $\frac k {k+1}$, the auctioneer gains revenue of $\frac{1}{2}$ in the $2$nd price auction for each signal. In contrast, consider the mixed signaling scheme with signals $\{0,j\}_{j=1,2,\ldots, k}$ where for every $j$, $\Pr[ \{0,j\} | \ \textrm{type }0] = \frac 1 k$ and $\Pr[ \{0,j\} | \ \textrm{type }j] = 1$. Now, for every signal, the auctioneer gains revenue of $\frac k {k+1}$.
\end{proof}

\section{Working Less by Sending Mixed Signals}
\label{sec:opt-mixed-singals-LP}

We now turn to showing that the mixed signaling scheme generating the largest revenue is polynomial-time computable. To that end, we construct a linear program, whose solution is the optimal signaling scheme. In order to devise the LP, we provide several observations, leading the way to formalization of the LP. But before proceeding to the LP and these observations, we introduce some notation.

\subsection{Notation.} 
\label{subsec:notation}

Given a signal $S$, we denote $w_1(S)$ as the bidder which is the winner of $S$ (the bidder with the highest bid), and $w_2(S)$ as the $2$nd highest bidder. Formally (recall that we identify $S$ and its support)
\begin{eqnarray*}
& w_1(S) & = \arg\max_i \{ \sum_{j\in S} \psi_{i,j}\varphi_{j,S} \} \cr
& w_2(S) & = \arg\max_{i\neq w_1(S)} \{ \sum_{j\in S} \psi_{i,j}\varphi_{j,S} \} = \arg\maxtwo_{i} \{ \sum_{j\in S} \psi_{i,j}\varphi_{j,S} \}
\end{eqnarray*} 
We call a signal $S$ a \emph{singleton} if $|S|=1$. Whenever $S$ is a singleton so $S=\{j\}$ for some $j$, we abbreviate $w_1(j), w_2(j)$. Given $i$, we denote by $d(i)$ the set of types for which the bid of $i$ is the biggest: $d(i) = \{j: w_1(j) = i\}$. Given a signal $S$, we denote $\pro(S)$ as the revenue the auctioneer gets from a $2$nd price auction over $S$. I.e., $\pro(S) = \sum_{j\in S} \varphi_{j,S}~ \psi_{w_2(S),j}$, the bid of bidder $w_2(S)$. As always, we abbreviate singletons to $\pro(j)$ instead of $\pro(\{j\})$. The overall revenue of the auctioneer from $\phi$ is defined as $\pro(\varphi) = \sum_{S\in\mathcal{S_\varphi}} \pro(S)$. We break ties arbitrarily, but in a consistent manner.

\subsection{Na\"ive LP.}
\label{subsec:naive_LP}

Our first observation is that the problem of finding the optimal signaling scheme can be formalized into an LP, with potentially many variables. Assume $\varphi^*$ is an optimal signaling scheme. We claim $\varphi^*$ has only a finite number of signals. The key observation is that the auctioneer has no need for two signals $S$ and $T$ s.t. $\supp(S) = \supp(T)$ and in both $S$ and $T$ bidder $i_1$ is the winning bidder ($w_1(S)=w_1(T)=i_1$) and bidder $i_2$ is the $2$nd highest bidder ($w_2(S) = w_2(T) = i_2$). We prove this observation rigorously.

\begin{claim}
\label{clm:altering_a_signaling_scheme}
Let $\varphi$ be a signaling scheme, and assume that there exist $S$ and $T\neq S$ s.t. $\supp(S) = \supp(T)$, and both $w_1(S) = w_1(T)$ and $w_2(S) = w_2(T)$. We define a new signaling scheme $\varphi'$ by ``merging'' $S$ and $T$ into a single signal $S'$, and keeping all other signals unchanged. Formally, let $\varphi'$ be the signaling scheme s.t.
\begin{eqnarray*}
 \forall j, && \varphi'(j,S') = \varphi(j,S) + \varphi(j,T), \qquad \varphi'(j,S) = \varphi'(j,T) = 0 \cr 
\forall j, S_0 \neq S, T, && \varphi'(j,S_0) = \varphi(j,S_0)
\end{eqnarray*}
Then for every bidder $i$ and item type $j$, the probability $i$ gets the item of type $j$ is identical in $\varphi$ and $\varphi'$, and $\pro(\varphi) = \pro(\varphi')$.
\end{claim}
\begin{proof}
The probability of $i$ winning item of type $j$ is exactly the probability that the auctioneer sees that the item is of type $j$ and then declares a signal $S$, for which $i$ has the winning bid. This clearly holds for all bidders but $w_1(S)$ ($=w_1(T)$), as all signals for which the winner isn't $w_1(S)$ are declared with the same probability in $\varphi$ and in $\varphi'$. The claim then follows from showing that $w_1(S)$ also has the winning bid for $S'$. 

First, observe that under the signal $S'$, the bidders bid $\E[v_i |\ S'] = \sum_j v_{i,j} \Pr[j |\ S'] = \frac {\sum_j p_j v_{i,j} \varphi'(j,S')} {\Pr[S']}$. Therefore, the order of the bids is determined by the numerator in the last term, as the denominator is the same for all bidders. By definition, for every $i$ we have that $\sum_j p_j v_{i,j} \varphi'(j,S') = \sum_j p_j v_{i,j} (\varphi(j,S)+\varphi(j,T))$, and so $w_1(S)$ had the winning bid for $S'$ and $w_2(S)$ has the second highest bid in $S'$. This allows us to deduce the first part of the claim.

As for revenue, it is evident that $\pro(\varphi')-\pro(\varphi) = \pro(S') - \left(\pro(S) + \pro(T)\right)$, and it is also simple to see that \[\pro(S') = \sum_{j\in S'} \varphi'(j,S')~ \psi_{w_2(S),j} = \sum_{j\in S} \varphi(j,S)~ \psi_{w_2(S),j} + \sum_{j\in T} \varphi(j,T)~ \psi_{w_2(T),j} = \pro(S) +\pro(T)\] because $S$ and $T$ have the same support as $T$, and because $w_2(S') = w_2(S) = w_2(T)$. We deduce $\pro(\varphi')-\pro(\varphi) = 0$.
\end{proof}

Following Claim~\ref{clm:altering_a_signaling_scheme}, it is evident that the number of signals in an optimal signaling scheme can be upper bounded by all possible subsets of types and pairs of bidders, so $|\mathcal{S}| \leq 2^m n^2$. Furthermore, constraining $i_1$ to the be the winning bidder and $i_2$ to be the second highest bidder for signal $S$, is simply a linear constraints. Therefore, by having a variable per signal and a pair of winning bidders, we get that the optimal signaling scheme is the solution for the following (exponential) LP:

\begin{equation}  \max \sum_{S\subset [m]}\sum_{i_1\neq i_2}\ \sum_{j\in S}  x_j(S,i_1, i_2) ~\psi_{i_2,j} \label{eq:naive_LP}\end{equation}
\begin{align}
 &\textrm{under constraints:}\cr
 & \forall S, \forall i_1\neq i_2, \ \forall i\neq i_1, i_2 & \sum_j   x_j(S,i_1, i_2)  ~\psi_{i_1,j} \geq \sum_{j\in S}   x_j(S,i_1, i_2)  ~\psi_{i,j} \cr
 &  & \sum_j   x_j(S,i_1, i_2)  ~\psi_{i_2,j} \geq \sum_{j\in S}   x_j(S,i_1, i_2)  ~\psi_{i,j} \label{constraint:i_1_and_i_2_are_best}\\
 & \forall S, \forall i_1\neq i_2, & \sum_j   x_j(S,i_1, i_2)  ~\psi_{i_1,j} \geq \sum_{j\in S}   x_j(S,i_1, i_2)  ~\psi_{i_2,j} \label{constraint:i_1_wins}\\
& \forall j, & \sum_{S: j\in S} \sum_{i_1\neq i_2}  x_j(S,i_1, i_2)  \leq 1 \cr
& \forall S, \forall i_1\neq i_2, & x_j(S,i_1, i_2) \geq 0 \notag 
\end{align}
Where the constraints in \eqref{constraint:i_1_and_i_2_are_best} assure $i_1$ and $i_2$ are the two highest bidders for $S$, and the constraint in \eqref{constraint:i_1_wins} assures $i_1$ wins for $S$. The last two constraints assure $\varphi$ indeed induces a probability for every $j$.

Therefore, our goal in the remainder of this section is to show that the number of variables in the LP \eqref{eq:naive_LP} can be reduced to a polynomial number. We comment that the same principals as in the proof of Claim~\ref{clm:altering_a_signaling_scheme} will be repeatedly applied in future claims. From now on, we omit the rigorous description of $\varphi'$, and merely refer to $\varphi'$ as the result of merging signals into a single signal / splitting a single signal into multiple signals.

\subsection{Reducing the Number of Variables in the LP.}
\label{subsec:improve_LP}

Our goal is to show that the number of subsets we need to consider in the abovementioned LP can be reduced to a number polynomial in $n$ and $m$. To show this, we follow a series of observations. 
In order to bound the number of signals needed, we'd ideally like to show that every signal can be split. That is, we would like to take any non-singleton signal $S$ in $\varphi^*$, and have the auctioneer declare a few signals of smaller support rather than declaring $S$. If such a thing is always possible, then we can recursively split signals until we're left with only singleton signals. 
\begin{definition}
\label{def:splittable_signal}
Given a signaling scheme $\varphi$, we call a signal $S\in \mathcal{S}_\varphi$ \emph{splittable} if there exists a partition $S = S_1 \cup S_2 \cup \ldots \cup S_t$ s.t. $\sum_{k=1}^t \pro(S_k) \geq \pro(S)$. We call a signal \emph{singleton-splittable} if the signal is splittable w.r.t the partition of the signal into $|S|$ singleton signals, that is $\sum_{j\in S} \pro(j) \geq \pro(S)$.
\end{definition}

Unfortunately, the existence of such a split is not always possible -- some signals are non-splittable. Our claims characterize exactly the cases where this split causes the auctioneer to lose revenue. 

\begin{claim}
\label{clm:same_winners}
Let $S$ be a signal in the optimal signaling scheme, which is \emph{not} singleton splittable. That is, $\pro(S) > \sum_{j\in S} \varphi_{j,S}~\pro(j)$. Then both $w_1(S)$ and $w_2(S)$ belong to the set of bidders that win the items of $S$: $\{w_1(j) : \ j\in S\}$.
\end{claim}
\begin{proof}
Assume that $w_1(S)$ does not belong to the set  $\{w_1(j) : \ j\in S\}$. It follows that for every $j$, the $2$nd highest bid cannot be smaller than the bid of the $w_1(S)$, and so we achieve the contradiction 
\[\sum_{j\in S} \varphi_{j,S}~\pro(j) \geq \sum_{j\in S} \varphi_{j,S}~\psi_{w_1(S),j} \geq \sum_{j\in S}\varphi_{j,S}~\psi_{w_2(S),j} = \pro(S)\] Similarly, if $w_2(S)$ isn't a winner for some $j\in S$, then for any $j$ we have that the bid of the bidder $w_2(j)$ is no less than the bid of $w_2(S)$. The inequality follows: $\sum_{j\in S}\varphi_{j,S}~\pro(j) \geq \sum_{j\in S}\varphi_{j,S}~\psi_{w_2(S),j} = \pro(S)$.
\end{proof}

The proof of Claim~\ref{clm:same_winners} gives the following as an immediate corollary.
\begin{corollary}
\label{cor:simply_splittable_signals}
Let $S$ be a signal s.t. the set $\{w_1(j) : \ j\in S\}$ contains a single bidder. Then $S$ is singleton-splittable.
\end{corollary}
\begin{proof}
If $S$ wasn't singleton-splittable, then the set $\{w_1(j) : \ j\in S\}$ would contain at least two distinct bidders.
\end{proof}

Corollary~\ref{cor:simply_splittable_signals} allows us to deduce that the non-splittable signals must contain at least two distinct bidders in their set of winners. We next show that non-splittable signals must contain at most two distinct bidders in this set.

\begin{claim}
\label{clm:two_winners}
There does not exist a non-splittable signal $S$ with $|\{w_1(j) : j\in S\}| \geq 3$.
\end{claim}
\begin{proof}
Assume the existence of a non-splittable signal $S$ with a set of winners, $\{w_1(j)\}_{j\in S}$, containing at least $3$ distinct bidders. From Claim~\ref{clm:same_winners} we know $w_1(S), w_2(S)$ belong to this set, and wlog we denote them simply as bidders $1 = w_1(S)$ and $2=w_2(S)$. This allows to denote $\pro(S) = \sum_{j\in S} \psi_{2,j} \varphi(j,S)$, where the winning bid for $S$ is $\sum_{j\in S} \psi_{1,j} \varphi(j,S)$. We now show $S$ is splittable.

Let us denote the following two disjoint subsets: $S_1 = S\cap d(1), S_2 = S\cap d(2)$, i.e., the set of types in $S$ that bidder $1$ (resp., bidder 2) covet the most. Observe that by assumption, some types in $S$ are not in $S_1\cup S_2$, so we can consider the partition $S = \big(S_1 \cup S_2\big) \cup \bigcup_{j \in S\setminus(S_1\cup S_2)} \{j\}$. I.e., we partition $S$ into $|S\setminus(S_1\cup S_2)| + 1$ signals: $|S\setminus(S_1\cup S_2)|$ singleton signals, and one signal for all types in $S_1 \cup S_2$.

First, we consider the revenue of the auctioneer from the singleton signals: $\sum_{j \notin S_1\cup S_2} \pro(j)$. On all such types $j$, neither bidder $1$ nor bidder $2$ have the highest bid, so the $2$nd highest bid is at least as high as the bid of bidder $1$ and the bid of bidder $2$. Therefore, on $S\setminus(S_1\cup S_2)$, the auctioneer's revenue is \[\sum_{j \notin S_1\cup S_2} \pro(j)\geq \max \{ \sum_{j\in S\setminus(S_1\cup S_2)} \varphi_{j,S} \psi_{1,j}, \sum_{j\in S\setminus(S_1\cup S_2)} \varphi_{j,S} \psi_{2,j}    \}\] 

Now we consider the revenue of the auctioneer from the signal $S_1\cup S_2$, where \emph{at least one} of the bidders $\{1,2\}$ doesn't have the winning bid. Therefore, the $2$nd highest bid is at least as high as the bid of bidder $1$ or the bid of bidder $2$. As a result, $\pro(S_1\cup S_2) \geq \sum_{j\in S_1\cup S_2} \varphi(j,S) \psi_{1,j}$ or $\pro(S_1\cup S_2) \geq \sum_{j\in S_1\cup S_2} \varphi(j,S) \psi_{2,j}$. It follows that the abovementioned partition of $S$ has revenue which is either $\pro(S_1\cup S_2) + \sum_{j\in S\setminus(S_1\cup S_2)} \pro(j) \geq \sum_{j\in S} \varphi(j,S) \psi_{1,j}$ or $\pro(S_1\cup S_2) + \sum_{j\in S\setminus(S_1\cup S_2)} \pro(j) \geq \sum_{j\in S} \varphi(j,S) \psi_{2,j}$. Observe that $\sum_{j\in S} \varphi(j,S) \psi_{1,j}$ is the winning bid of $S$, so $\sum_{j\in S} \varphi(j,S) \psi_{1,j} \geq \sum_{j\in S} \varphi(j,S) \psi_{2,j} = \pro(S)$, and deduce that in any case $\pro(S_1\cup S_2) + \sum_{j\in S\setminus(S_1\cup S_2)} \pro(j) \geq \pro(S)$. Contradiction.
\end{proof}

Combining Corollary~\ref{cor:simply_splittable_signals} and Claim~\ref{clm:two_winners} we deduce the following.
\begin{corollary}
\label{cor:support_containment}
Let $S$ be a non-splittable signal in $\varphi^*$. Then $\{w_1(j) :\ j\in S\} = \{w_1(S), w_2(S)\}$, otherwise denoted as $\supp(S) \subset d(w_1(S)) \cup d(w_2(S))$.
\end{corollary}

Using Corollary~\ref{cor:support_containment} we deduce the existence of an optimal signaling scheme with exactly two types of signals: either singleton signals, or non-splittable signals. Now, using Claim~\ref{clm:altering_a_signaling_scheme} we can take any two non-splittable signals $S,T$ such that $w_1(S)=w_1(T)$ and that $w_2(S) = w_2(T)$ and merge them. This follows from the fact that we can always think of $S$ and $T$ as two signals over $d(w_1(S)) \cup d(w_2(S))$, with some elements have $0$ probability of declaring $S$ (or $T$).

Using~\ref{clm:same_winners}, \ref{cor:support_containment} and , we deduce that there exists a signaling scheme that has at most  $m + n(n-1)$ different signals: the singleton signals, and the signals composed from pairing $d(i)$ and $d(i')$ for any two bidders $i, i'$. Observe that $d(1), d(2), \ldots, d(n)$ partition the $m$ different types into disjoint sets, so there can only be $\min\{m,n\}$ such elements in the partition. We therefore deduce that the optimal signaling scheme has at most $N = m + \min\{m(m-1), n(n-1) \} \leq m^2$ signals. We can therefore reduce our LP to have $N$ variables: variables $x_j$, indicating the probability that the auctioneer sees item of type $j$ and declares the singleton cluster $\{j\}$; and variables $y_j(i_1, i_2)$, indicating the probability that the auctioneer sees item of type $j \in d(i_1)\cup d(i_2)$ and declares a signal in which $i_1$ has the highest bid, and $i_2$ has the second highest bid. Formally, we solve:
\begin{equation}  \max \sum_j x_j ~\psi_{w_2(j),j} + \sum_{i_1}\sum_{i_2\neq i_1} \sum_{j \in d(i_1)\cup d(i_2)} y_j(i_1, i_2) ~\psi_{i_2,j} \label{eq:LP}\end{equation}
\begin{align*}
 &\textrm{under constraints:}\\
 & \forall i_1,\ \forall i_2\neq i_1, & \sum_{j\in d(i_1)\cup d(i_2)} y_{j}(i_1,i_2)~\psi_{i_1,j} \geq \sum_{j\in d(i_1)\cup d(i_2)} y_{j}(i_1,i_2) ~\psi_{i_2,j}\\
& \forall i_1,\ \forall i_2\neq i_1,\textrm{ and } i \neq i_1, i_2, & \sum_{j\in d(i_1)\cup d(i_2)} y_{j}(i_1,i_2) ~\psi_{i_1,j} \geq \sum_{j\in d(i_1)\cup d(i_2)} y_{j}(i_1,i_2) ~\psi_{i,j} \\
&& \sum_{j\in d(i_1)\cup d(i_2)} y_{j}(i_1,i_2) ~\psi_{i_2,j} \geq \sum_{j\in d(i_1)\cup d(i_2)} y_{j}(i_1,i_2) ~\psi_{i,j} \\
& \forall j, & x_{j} +  \sum_{i_1}\sum_{\substack{i_2\neq i_1 \\ \textrm{s.t. } j\in d(i_1)\cup d(i_2)}} y_{j}(i_1,i_2) \leq 1 \\
&  \forall j,\ \forall i_1,\ \forall i_2\neq i_1, & x_{j} \geq 0, \qquad  y_{j}(i_1, i_2) \geq 0
\end{align*}

\subsection{An Additional Observation} 
\label{subsec:additional_observation}

Note that for every $i_1\neq i_2$ and every $j\in d(i_1)\cup d(i_2)$, we have two $y$-variables in the LP~\eqref{eq:LP}, one for $i_1$ winning and $i_2$ coming second, and one for $i_2$ winning and $i_1$. We now show that it is enough to use just one variable, indicating a signal in which \emph{both} $i_1$ and $i_2$ give the highest bid.

\begin{observation}
\label{obs:equal_first_and_second_bid}
There exists an optimal signaling scheme, in which for each non-singleton signal $S$, the first and the second highest bid are identical.
\end{observation}
\begin{proof}
Assume that for a certain signal $S$, the bid of $w_1(S)$ is strictly greater than the bid of $w_2(S)$. Wlog, denote bidder $1$ as $w_1(S)$ and bidder $2$ as $w_2(S)$. We split $S$ into two disjoint, non-empty sets $S_1 = S \cap d(1)$ and $S_2 = S\cap d(2)$. (If either $S_1$ or $S_2$ are empty, then Corollary~\ref{cor:simply_splittable_signals} shows $S$ can be split into singleton signals.) Define \[g = \frac {\sum_{j\in S_2} \varphi_{j,S}~(\psi_{2,j} -\psi_{1,j}) } {\sum_{j\in S_1} \varphi_{j,S}~(\psi_{1,j} -\psi_{2,j}) }\] (Note, both the numerator and the denominator or positive.) By assumption, we have
\begin{align*}
& \sum_{j\in S} \varphi(j,S) ~\psi_{1,j} > \sum_{j\in S} \varphi(j,S) ~\psi_{2,j} && \Leftrightarrow \cr
& \sum_{j\in S_1} \varphi_{j,S}~\psi_{1,j} +  \sum_{j\in S_2} \varphi_{j,S}~\psi_{1,j} > \sum_{j\in S_2} \varphi_{j,S}~\psi_{2,j} + \sum_{j\in S_2} \varphi_{j,S}~\psi_{2,j} && \Leftrightarrow \cr
& \sum_{j\in S_1} \varphi_{j,S}~(\psi_{1,j} -\psi_{2,j}) >  \sum_{j\in S_2} \varphi_{j,S}~(\psi_{2,j} -\psi_{1,j}) && \Leftrightarrow 
\ \ g < 1
 \end{align*}
So now, define $\varphi'$ to be the signaling scheme where for any $j\in S_1$, the probability of giving the signal $S$ decreases: $\varphi'(j,S) = g\cdot \varphi_{j,S}$, and as a result, the probability of giving the singleton signal $\{j\}$ increases: $\varphi'(j, \{j\}) = \varphi_{j, \{j\}} + (1-g)\cdot \varphi_{j,S}$. In $\varphi'$, the above derivation shows that the bid of bidder $1$ and of bidder $2$ are identical. Furthermore, by increasing the probability mass on the singleton signals, the auctioneer can only increase her revenue.
\end{proof}

Following Observation~\ref{obs:equal_first_and_second_bid}, we deduce that the number of variables in the LP (and the number of signals in our signaling scheme) can be bounded by $N = m + \min\{\binom{n}{2}, \binom{m}{2}\}$. Furthermore, Observation~\ref{obs:equal_first_and_second_bid} justifies the fact that we repeatedly identify a signal with its support.

\section{Competitiveness Against a Benchmark}
\label{sec:approximating-benchmark}

%Ideally, we would like to end this work by showing that the signaling scheme we devise obtains optimal / close-to-optimal revenue out of all possible mechanisms. Unfortunately, it is unclear what is the optimal mechanism (the mechanism that extracts the most revenue). 
We show a lower bound for the revenue against a benchmark, and first discuss which benchmarks are reasonable.
It is quite clear, especially when viewed as selling $m$ divisible goods, that the auctioneer cannot get more than $\sum_j \max_i \psi_{i,j}$. As we are restricted to run a 2nd price auction in the end, this quantity is in general unapproachable, since some bidder might have valuations that are so high that they overshadow all other valuations of all other bidders. We thus define our benchmark as the outcome of ``taking a bidder out of the picture''. That is, we ignore the bids of $i'$, and sum the maximum bid for each type separately. Formally,
\[\B = \min_{i'} \left(\sum_j \max_{i\neq i'} \psi_{i,j}\right) \ \ = \min_{i'} \left(\sum_{j\in d(i')} \maxtwo_i \psi_{i,j} + \sum_{j\notin d(i')} \max_{i} \psi_{i,j}\right) \]

Before showing our algorithm is competitive with $\mathcal{B}$, let us first discuss the motivation for this benchmark.  A classical benchmark for comparison in other prior-free settings is the one the results from omitting the bidder with the highest bid (for the same reasoning mentioned above), see Goldberg {\em et al.} \cite{Goldberg}. Therefore, one might suggest that the right benchmark for the problem is result of ignoring the bid of the one bidder who covets her set of item types the most. Formally, this other benchmark for the problem is: $\tilde \B = \sum_j (\max_{i\neq i^*} \psi_{i,j})$ where $i^* = \arg \max_i \left(\sum_{j\in d(i)} \psi_{i,j}\right)$. First, observe that $i^*$ and $i_0$, the bidder for which the benchmark $\B$ is obtained, are not necessarily the same, as the example in Figure~\ref{fig:example} demonstrate. But let us show that are closely related. 
\begin{figure}
\begin{center}
\fbox{\parbox{2.5in}{
\begin{tabular*}{2.5in}{r|c|c|c}
 & type $1$ & type $2$ & type $3$ \cr
\hline 
bidder $1$ & 500 & 500 & 0 \cr 
bidder $2$ & 499 & 498 & 1 \cr 
bidder $3$ & 7 & 3 & 999 \cr 
\end{tabular*}
}}
\caption{\label{fig:example} \scriptsize An example demonstrating that $\B$ and $\tilde\B$ are different. $\tilde\B$ requires we ignore bidder $1$, as her bids are the highest. $\B$ requires we ignore bidder $3$, as ignoring bidder $3$ results in the biggest decrease in maximal bids.}
\end{center}
\end{figure}

\begin{claim}
\label{clm:comparing_benchmarks}
\[ \tilde\B/2 \leq \B \leq \tilde\B \]
\end{claim}
\begin{proof}
The inequality $\B \leq \tilde\B$ follows from the definition of $\B$, so we turn to proving $\B \geq \frac 1 2 \tilde\B$. Denote $i_0$ as the bidder on which the minimum of $\B$ is obtained. Obviously, if $i^* = i_0$, we are done, as both benchmarks are the same. So we assume $i^*\neq i_0$, and we have
\begin{eqnarray*}
&&\B = \sum_{j} (\max_i \psi_{i,j}) - \left(\sum_{j\in d(i_0)} (\max_i \psi_{i,j}) - (\maxtwo_i \psi_{i,j})\right) \geq \sum_{j\in d(i^*)} (\max_i \psi_{i,j}) \cr
&&\tilde\B = \sum_{j} (\max_i \psi_{i,j}) - \left(\sum_{j\in d(i^*)} (\max_i \psi_{i,j}) - (\maxtwo_i \psi_{i,j})\right) \cr
\end{eqnarray*}
and since, by definition of  $i^*$, we have that $\sum_{j\in d(i^*)} \psi_{i,j} \geq \sum_{j\in d(i_0)} \psi_{i,j}$ then it holds that
\begin{eqnarray*}
&\tilde\B - \B &= \left(\sum_{j\in d(i_0)} (\max_i \psi_{i,j}) - (\maxtwo_i \psi_{i,j})\right) - \left(\sum_{j\in d(i^*)} (\max_i \psi_{i,j}) - (\maxtwo_i \psi_{i,j}) \right)\cr
&& \leq \sum_{j\in d(i^*)} (\maxtwo_i \psi_{i,j}) + \left( \sum_{j\in d(i_0)} (\max_i \psi_{i,j}) - \sum_{j\in d(i^*)} (\max_i \psi_{i,j})  \right) \cr
&& \leq \sum_{j\in d(i^*)} (\maxtwo_i \psi_{i,j}) \leq \sum_{j\in d(i^*)} (\max_i \psi_{i,j}) \leq \B %\qedhere
\end{eqnarray*}
\end{proof}

We comment that the example in Figure~\ref{fig:example} also demonstrates that the $2$-factor of Claim~\ref{clm:comparing_benchmarks} is essentially tight. Now, having established the connection between $\B$ and $\tilde\B$, we compare our signaling scheme with the benchmark $\B$.

\begin{theorem}
\label{thm:2-approximation}
For any set of valuations $\psi_{i,j}$, the revenue of our signaling scheme is $\geq \B/2$.
\end{theorem}
\begin{proof}
The proof follows from breaking the revenue of the signaling scheme into two terms: the revenue from singleton signals, and the revenue from non-singleton signals. Given a signaling scheme $\varphi$, we denote
\begin{eqnarray*}
&&\pro^{\rm S}(\varphi) = \sum_j \pro(j) = \sum_j \varphi(j,\{j\}) ~\psi_{w_2(j),j} \cr
&&\pro^{\rm NS}(\varphi) = \sum_{S:\ |S|\geq 2} \pro(S) = \sum_{S:\ |S|\geq 2}  \varphi(j,S)\sum_{j\in S} \psi_{w_2(S),j}
\end{eqnarray*}
We now denote $\varphi^*$ as the optimal signaling scheme we get from solving the LP in (\ref{eq:LP}).
Let's fix $S$ to be some non-singleton signal in our scheme. So $S$ corresponds to a pair of bidders, $i$ and $i'$, such that $S\subset d(i)\cup d(i')$ and $i$ has the highest bid on $d(i)$, whereas $i'$ has the highest bid over the items in $d(i')$. The revenue the auctioneer gets from signal $S$ is exactly $\pro(S) = \sum_{j\in d(i)\cup d(i')} y_{j}(i,i') \psi_{i,j} = \sum_{j\in d(i)\cup d(i')} y_{j}(i,i') \psi_{i',j}$. Therefore, \[2\cdot\pro(S) = \sum_{j\in d(i)\cup d(i')} y_{j}(i,i') (\psi_{i,j} + \psi_{i',j}) \geq \sum_{j\in d(i)\cup d(i')} y_{j}(i,i')(\max_i\psi_{i,j})\] Summing up the revenue of the auctioneer from all non-singleton signals, we have
\begin{eqnarray}
\label{eq:profit_non_singletons}
& \pro^{\rm NS}(\varphi^*) & \geq \frac 1 2 \sum_{S:\ |S|\geq 2} \sum_{j\in S} (\max_i\psi_{i,j})y_{j}(i,i') = \frac 1 2 \sum_j (\max_i\psi_{i,j}) \sum_{\substack{S:\ |S|\geq 2\\ j\in S}} \varphi^*(j,S) \cr 
&& = \frac 1 2 \sum_j (\max_i \psi_{i,j}) \Pr[\textrm{Given }j,\ \varphi^*\textrm{ declares a non-singleton signal}] \cr 
&& = \frac 1 2 \sum_j (\max_i \psi_{i,j})\big(1-\varphi^*(j,\{j\})\big)
\end{eqnarray}
We now turn to bound the revenue from singleton signals, that is, the term $\pro^{\rm S}(\varphi^*)$. Let us consider the following procedure, that converts one signaling scheme $\varphi$ into a different one $\varphi'$.

\begin{enumerate}
\item Let $j = \arg\min \{\varphi(j,\{j\}) ~\psi_{w_1(j),j} : \ \varphi(j,\{j\})>0\}$. 
\item Fix some $j'$ s.t. $\varphi(j',\{j'\}) >0$ and s.t. $w_1(j) \neq w_1(j')$.
\item Define $\lambda = \frac{ \varphi(j,\{j\}) ~ \psi_{w_1(j),j} } {\varphi(j',\{j'\}) ~ \psi_{w_1(j'),j'}}~$ (obviously, $\lambda \leq 1$).
\item Alter $\varphi$ in the following manner. Introduce a new signal $S_{\rm new} = \{j, j'\}$ and set
\begin{align*}
&\varphi'(j,S_{\rm new}) = \varphi(j,\{j\}) && \varphi'(j,\{j\}) = 0 \\
&\varphi'(j',S_{\rm new}) = \lambda \varphi(j',\{j'\}) &&\varphi'(j',\{j'\}) = (1-\lambda)\varphi(j',\{j'\})
\end{align*}
\end{enumerate}

Now, the effect of applying this procedure on a signaling scheme is that $\pro^{\rm S}$ decreases, yet $\pro^{\rm NS}$ increases: $\pro^{\rm S}(\varphi') - \pro^{\rm S}(\varphi) = - \varphi_{j,j} ~\psi_{w_2(j),j} - \lambda\varphi_{j',j'} ~\psi_{w_2(j'), j'}$, whereas $\pro^{\rm NS}(\varphi') - \pro^{\rm NS}(\varphi) = \pro(S_{\rm new})$. But now, because of $\lambda$, the bids of $w_1(j)$ and the bid of $w_1(j')$ are identical for $S_{\rm new}$, and therefore, just as shown above, $\pro(S_{\rm new}) = \frac 1 2 \big(\varphi(j,\{j\}) ~\psi_{w_1(j),j} + \lambda\varphi(j',\{j'\})~\psi_{w_1(j'),j'}\big)$.

Given a signaling scheme, we denote $J_\varphi = \{j: \ \varphi(j,\{j\})>0\}$, and $I_\varphi = \{w_1(j) : j\in J_\varphi\}$. It is evident that the above  procedure is applicable as long as $I$ contains at least two distinct bidders. So, imagine we take $\varphi^*$ and apply the abovementioned procedure repeatedly, until it is no longer applicable. (Note, every time we apply the procedure, we decrease the number of singleton signals by at least $1$, so in $m$ iterations we must terminate.) Denote the signaling scheme which we end with by $\bar\varphi$, and assume $I_{\bar\varphi} $ contains a single bidder, $i_0$ (the case $I_{\bar\varphi} =\emptyset$ is even simpler). Denote $J_{\rm remain}$ as all the types that appear as singleton in $\bar\varphi$ (and obviously in $\varphi^*$), and observe that $J_{\rm remain}\subset d(i_0)$.

Repeating the derivation from \eqref{eq:profit_non_singletons}, we get that
\begin{eqnarray*}
& \pro^{\rm NS}(\bar\varphi) & \geq \frac 1 2 \sum_j (\max_i \psi_{i,j})\big(1-\bar\varphi(j,\{j\})\big) \cr 
&& = \frac 1 2 \sum_{j\notin J_{\rm remain}} (\max_i \psi_{i,j}) + \frac 1 2 \sum_{j\in J_{\rm remain}} (\max_i \psi_{i,j})\big(1-\bar\varphi(j,\{j\})\big) \cr 
&& \geq \frac 1 2 \sum_{j\notin J_{\rm remain}} (\max_i \psi_{i,j}) + \frac 1 2 \sum_{j\in J_{\rm remain}} (\maxtwo_i \psi_{i,j})\big(1-\bar\varphi(j,\{j\})\big)  
\end{eqnarray*}

Since $\pro^{\rm S}(\bar\varphi) = \sum_{j\in J_{\rm remain}} \bar\varphi(j, \{j\}) (\maxtwo_i \psi_{i,j})$ we can conclude and deduce that
\begin{eqnarray*}
& \pro(\varphi^*) \geq \pro(\bar\varphi)  &=  \pro^{\rm NS}(\bar\varphi) + \pro^{\rm S}(\bar\varphi) \cr 
&& \geq \frac 1 2 \sum_{j\notin J_{\rm remain}} (\max_i \psi_{i,j}) + \frac 1 2 \sum_{j\in J_{\rm remain}} (\maxtwo_i \psi_{i,j}) \cr 
&& \geq \frac 1 2 \sum_{j\notin d(i_0)} (\max_i \psi_{i,j}) + \frac 1 2 \sum_{j\in d(i_0)} (\maxtwo_i \psi_{i,j}) \cr 
&& = \frac 1 2 \sum_{j} (\max_{i\neq i_0} \psi_{i,j}) \geq \B %\qedhere
\end{eqnarray*}
\end{proof}

We comment that if $\varphi^*$ or $\bar\varphi$ contains no singleton signals, then we have $\pro(\varphi^*)\geq \pro(\bar\varphi) = \pro^{\rm NS}(\bar\varphi) = \frac 1 2 \sum_j \max_{i} \psi_{i,j}$. In words: if the optimal signaling scheme contains no singleton clusters, then the revenue of the auctioneer is at least half the sum of highest bids over all types. We also comment that the example in the introduction, the one where $m=n$ and the valuations form the unit matrix, exhibit a case where the $2$-factor in Theorem~\ref{thm:2-approximation} is tight.

\section{Discussion and Open Problems}
\label{sec:conclusion}

We have shown that in probabilistic single item auctions, mixed signaling schemes outperforms pure ones, both with respect to revenue and with respect to computational complexity. Furthermore, Observation~\ref{obs:equal_first_and_second_bid} gives us an insight as to the characterization of the optimal signaling / bundling scheme. The auctioneer leverages her informational advantage to bundle goods in a way that \emph{maximizes competition} among bidders -- her non-singleton bundles are exactly those where two (or more) bidders are equal in their utility. In that aspect, our model allows us to \emph{precisely} quantify the extent for which the seller can shape the demand in order to increase her revenue (rather than the usual concern of truthfully sampling the demand, in the non-full information setting). Needless to say, the notion that an increase in the demand leads to an increase in revenue is a basic principle of microeconomics (e.g.~\cite{mascolell1995mt}).

Similarly, Observation~\ref{obs:equal_first_and_second_bid} also demonstrates the connection between our signaling scheme and the fractional knapsack problem (see~\cite{KelPfePis04}). In fact, one may view the problem as a version of the knapsack problem -- for every pair of bidders $(i,i')$ there are numerous ways of bundling the goods s.t. the bids of $i$ and $i'$ are the same. The auctioneer is therefore faced with the problem of picking a subset of these potential bundles (subject to having at most one unit of each good) in order to maximize her profit. And, much like the fact that the fractional knapsack problem is polynomial time solvable, so is the mixed signals problem.

Finally, we suggest some interesting open problems:
\begin{itemize}
\item{}Are there instances where the optimal mixed signaling scheme generates strictly more than twice the revenue of the optimal pure signaling schemes?
\item{}In Bayesian variants of the setup (see \cite{Emek}), how well does the signaling + 2nd price auction approach approximate the optimal auction (in the sense of Myerson \cite{Myer})? 
\item{}Is it possible to find an optimal  (or approximately optimal) signaling scheme when $m$ is exponentially large? Consider the case where each type can be described using $d$ attributes, and the bidders' valuations for the item are functions of these $d$ attributes. Can one extend the LP of~\eqref{eq:LP} to handle such valuations?
\end{itemize}
\bibliography{mixed}

\newcommand{\etalchar}[1]{$^{#1}$}
\begin{thebibliography}{MCWGdCEiE95}

\bibitem[AC04]{Ausubel04auctioningmany}
Lawrence~M. Ausubel and Peter Cramton.
\newblock Auctioning many divisible goods.
\newblock {\em Journal of the European Economics Association}, 2:2--3, 2004.

\bibitem[Ake70]{akerlof1970}
George~A Akerlof.
\newblock The market for 'lemons': Quality uncertainty and the market
  mechanism.
\newblock {\em The Quarterly Journal of Economics}, 84(3):488--500, August
  1970.

\bibitem[AY76]{Adam}
William~James Adams and Janet~L. Yellen.
\newblock Commodity bundling and the burden of monopoly.
\newblock {\em The Quarterly Journal of Economics}, 90(3):475--498, August
  1976.

\bibitem[BZ01]{Back_Zender_2001}
Kerry Back and Jaime~F Zender.
\newblock Auctions of divisible goods with endogenous supply.
\newblock {\em Economics Letters}, 73(1):29--34, 2001.

\bibitem[CSSS10]{cramton2010combinatorial}
P.~Cramton, Y.~Shoham, R.~Steinberg, and V.L. Smith.
\newblock {\em Combinatorial Auctions}.
\newblock MIT Press, 2010.

\bibitem[EFGT11]{Emek}
Yuval Emek, Michal Feldman, Iftah Gamzu, and Moshe Tennenholtz.
\newblock Revenue maximization in probabilistic single-item auctions via
  signaling.
\newblock In {\em Seventh Ad Auctions Workshop. Online proceedings at
  http://sites.google.com/site/adauctions2011/program-and-schedule}, 2011.

\bibitem[GHK{\etalchar{+}}06]{Goldberg}
Andrew~V. Goldberg, Jason~D. Hartline, Anna~R. Karlin, Michael Saks, and Andrew
  Wright.
\newblock Competitive auctions.
\newblock {\em Games and Economic Behavior}, 55(2):242 -- 269, 2006.

\bibitem[IK08]{Iyengar_Kumar_2008}
Garud Iyengar and Anuj Kumar.
\newblock Optimal procurement auctions for divisible goods with capacitated
  suppliers.
\newblock {\em Review of Economic Design}, 12(2):129–154, 2008.

\bibitem[Kle04]{klemperer2004auctions}
P.~Klemperer.
\newblock {\em Auctions: theory and practice}.
\newblock Toulouse lectures in economics. Princeton University Press, 2004.

\bibitem[KPP04]{KelPfePis04}
H.~Kellerer, U.~Pferschy, and D.~Pisinger.
\newblock {\em Knapsack Problems}.
\newblock Springer, Berlin, Germany, 2004.

\bibitem[MCWGdCEiE95]{mascolell1995mt}
A.~Mas-Colell, M.D. Whinston, J.R. Green, and Universitat Pompeu Fabra~Facultat
  de~Ci{\`e}ncies Econ{\`o}miques~i Empresarials.
\newblock {\em {Microeconomic theory}}.
\newblock Oxford University Press New York, 1995.

\bibitem[Mye81]{Myer}
Roger~B. Myerson.
\newblock Optimal auction design.
\newblock {\em Mathematics of Operations Research}, 6(1):58--73, February 1981.

\end{thebibliography}
\end{document}